%% file: main.tex
\def\year{2022}\relax
\documentclass[letterpaper]{article} 
\usepackage{aaai22}  
\usepackage{times}  
\usepackage{helvet}  
\usepackage{courier}  
\usepackage[hyphens]{url}  
\usepackage{graphicx} 
\usepackage{natbib}  
\usepackage{caption} 
\DeclareCaptionStyle{ruled}{labelfont=normalfont,labelsep=colon,strut=off} 
\usepackage{amsmath,amssymb,amsfonts,amsthm}
\usepackage{stmaryrd}
\usepackage{textcomp}
\usepackage{algorithm}
\usepackage{algorithmic}
\usepackage{balance}
\usepackage[x11names]{xcolor}
\usepackage{multirow}
\usepackage{arydshln}
\usepackage[colorlinks=true,urlcolor=blue]{hyperref}

\newtheorem{theorem}{\textbf{Theorem}}
\newtheorem{definition}{\textbf{Definition}}

\newtheorem{proposition}{\textbf{Proposition}}
\newtheorem{remark}{\textbf{Remark}}

\usepackage{newfloat}
\usepackage{listings}
\floatstyle{ruled}
\newfloat{listing}{tb}{lst}{}
\floatname{listing}{Listing}
\title{Fine-grained Private Knowledge Distillation}
\author {
    Yuntong Li \textsuperscript{\rm 1},
    Shaowei Wang \textsuperscript{\rm 1},
    Yingying Wang \textsuperscript{\rm 1},
    Jin Li \textsuperscript{\rm 1} \\
    Yuqiu Qian \textsuperscript{\rm 2},
    Bangzhou Xin \textsuperscript{\rm 3},
    Wei Yang\textsuperscript{\rm 3}
}
\affiliations {
    \textsuperscript{\rm 1} Institute of Artificial Intelligence and Blockchain,
Guangzhou University, Guangzhou, China.\\
    \textsuperscript{\rm 2} Interactive Entertainment Group, Tencent Games, Shenzhen, China.\\
    \textsuperscript{\rm 3} Institute of Advanced Research,
     University of Science and Technology of China, Suzhou, China.\\
}

\usepackage{bibentry}

\begin{document}

\maketitle

\begin{abstract}
Knowledge distillation has emerged as a scalable and effective way for privacy-preserving machine learning. One remaining drawback is that it consumes privacy in a model-level (i.e., client-level) manner, and every distillation query incurs privacy loss of one client's all records. In order to attain fine-grained privacy accountant and improve utility, this work proposes a model-free \textit{reverse $k$-NN labeling} method towards record-level private knowledge distillation, where each record is employed for labeling at most $k$ queries. Theoretically, we provide bounds of labeling error rate under the centralized/local/shuffle model of differential privacy (w.r.t. the number of records per query, privacy budgets). Experimentally, we demonstrate that it achieves new state-of-the-art accuracy with one order of magnitude lower of privacy loss. Specifically, on the CIFAR-$10$ dataset, it reaches $82.1\%$ test accuracy with centralized privacy budget $1.0$; on the MNIST/SVHN dataset, it reaches $99.1\%$/$95.6\%$ accuracy respectively with budget $0.1$. It is the first time deep learning with differential privacy achieves comparable accuracy with reasonable data privacy protection (i.e., $\exp(\epsilon)\leq 1.5$). Our code is available at \href{https://github.com/liyuntong9/rknn}{https://github.com/liyuntong9/rknn}.
\end{abstract}


\input{introduction}

\input{relatedwork}

\input{preliminary}

\input{nonprivate}
\input{central}

\input{local}
\input{shuffle}
\input{experiment}
\balance
\input{conclusion}

\newpage
\bibliography{refs}
\newpage
\input{appendix}

\end{document}

%% file: introduction.tex
\section{Introduction}\label{sec:intro}
\noindent Federated machine learning benefits from data across multiple individuals or organizations. However, data privacy has been a critical issue during collaboration, especially under increasingly rigid privacy laws, such as General Data Protection Regulation in the Europe Union, California Consumer Privacy Acts in California, and Data Security Law of the PRC in China. In contrast to transmitting raw data among clients, the seminal work of federated learning \cite{konevcny2016federated} proposes to share gradients. Subsequent works \cite{wei2020federated,luo2021scalable} further impose rigorous protections (e.g., differential privacy \cite{dwork2006differential}) on the gradients. Since iteratively transmitting gradients is inefficient (especially for deep neural networks), researchers \cite{papernot2016semi,papernot2018scalable} begin to employ the paradigm of knowledge distillation \cite{hinton2015distilling}. Federated clients are asked to label public-available data with locally-trained models (see the top of Figure \ref{fig:comparison}), meanwhile preserving the privacy of clients' local records. Because labels have much lower dimensionality than gradients, federated knowledge distillation has become a communication\&privacy efficient and thus prevalent way to federated deep learning \cite{lyu2020differentially,zhu2021data}. 

\begin{figure}
\begin{center}
\centerline{\includegraphics[width=85mm]{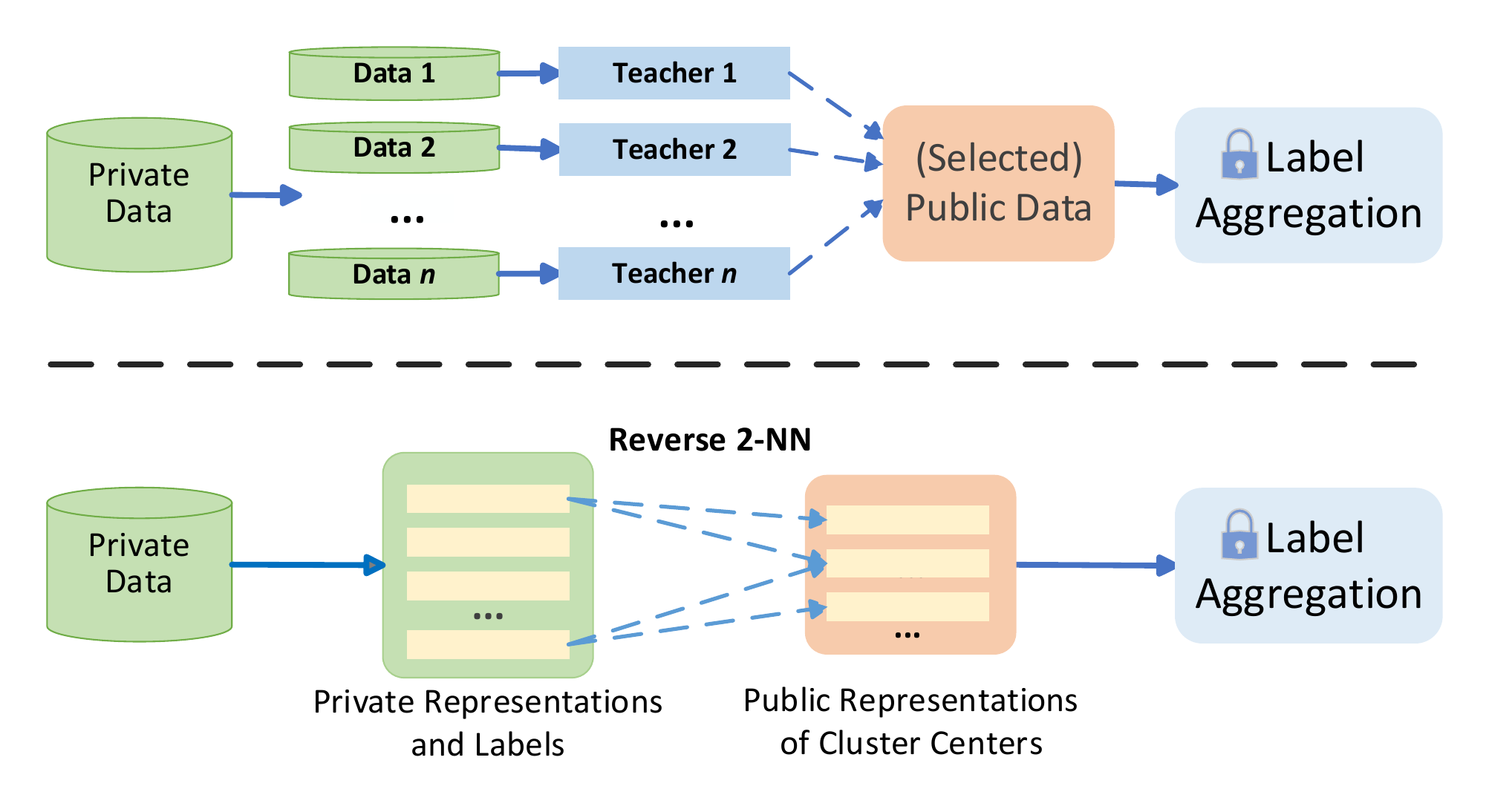}}
\vskip -0.05in
\caption{Comparison of the current paradigm of federated knowledge distillation (top) and our record-level private knowledge distillation with reverse $k$-NN (bottom).}
\label{fig:comparison}
\end{center}
\vspace*{-1.5em}
\end{figure}

Despite many advantages over sharing gradients, the current knowledge distillation paradigm is still suffering from a critical drawback on privacy accountant. Instead of accounting privacy loss at the record level as the gradient-sharing paradigm, one client's records in the knowledge-distillation paradigm are summarized to a privacy-sensitive local model. When answering a distillation query with the local model, it is then almost intractable to account for each record's contribution to the answer, or to concisely derive each record's privacy loss.

Such a coarse client-level privacy consumption wastes privacy budget and fails to extract maximum knowledge from each record. Therefore, this work initializes the study of federated knowledge distillation with record-level (pure) differential privacy. With the help of unsupervised representation learning, we propose the model-free \textit{reverse $k$-NN labeling} to achieve bounded contribution and constant privacy loss of every record. As demonstrated in Figure \ref{fig:comparison}, every record is associated with $k$ nearest neighboring querying samples, hence its privacy loss scales with only $k$ instead of the number of total queries. In concise, the advantages of this novel method over the current paradigm of federated knowledge distillation are two-fold:
\begin{itemize}
    \item{\textbf{Fine-grained privacy accountant. \ } As opposed to the current paradigm that every record's contribution and privacy loss (through local model) can not be separately accounted for, the new method limits one single record's contribution, which enables more fine-grained privacy accountants and better knowledge elicitation.}
    \item{\textbf{Broad application scenarios. \ } Instead of relying on a locally trained model that needs hundreds of training records, our method only needs an appropriate distance measure (e.g., via learned representations) for reverse $k$-NN. Therefore, it could apply to both the cross-silo (i.e., every client has relatively abundant records) and cross-device (i.e., every client has only a few records or one record) settings in federated learning. Besides, it is naturally immune to data \emph{Non-I.I.D.} settings and incurs only few-rounds communication.}
\end{itemize}

\subsection{Contributions}
This work formulates the model-free reverse $k$-NN labeling problem as \textit{Bucketized Sparse Vector Summation}, and then provides thorough solutions for the problem under centralized/local/shuffle differential privacy \cite{dwork2006differential}. Theoretically, We also analyze the labeling error rate of proposed solutions. The contributions of this paper are as follows.
\begin{itemize}
    \item{We initialize the study of federated knowledge distillation with record-level privacy preservation, and propose the model-free reverse $k$-NN query labeling method for achieving record-level (pure) differential privacy.}
    \item{We formulate the model-free reverse $k$-NN labeling problem as Bucketized Sparse Vector Summation (i.e., BSVS), and provide concrete mechanisms \& theoretical guarantees for the problem under centralized/local/shuffle differential privacy.}
    \item{For the first time, we show that the most stringent scenario of local private federated deep learning (with knowledge distillation) is practical. It reaches $98.5\%$ test accuracy on MNIST and $78.2\%$ test accuracy on CIFAR-10 with a reasonable local budget $\epsilon=0.4$.}
    \item{Through experiments, we demonstrate that our method achieves a significant accuracy boost meanwhile consuming an order of magnitude less of privacy budget when compared to existing approaches.}
\end{itemize}

%% file: relatedwork.tex
\section{Related Work}\label{sec:related}
We here retrospect efforts toward private and accurate knowledge distillation. The seminal work of knowledge distillation \cite{hinton2015distilling} transfers knowledge from a large teacher model to a compact student model, and aims for boosting inference efficiency and accuracy simultaneously. Nowadays researches on decentralized/federated learning \cite{papernot2016semi,papernot2018scalable,lin2020ensemble} employ ensemble knowledge distillation as a communication\&privacy efficient learning paradigm beyond gradient/parameter aggregation \cite{abadi2016deep}. 

Some studies (e.g., in \cite{sun2020federated}) discover that knowledge distillation naturally resists membership inference attacks to some extent. For formally guaranteeing data privacy, it is necessary to conform to the centralized/local/shuffle differential privacy (DP) during distillation. Specifically, one line of studies inject Laplace/Gaussian random noise for preserving centralized $\epsilon$-DP on aggregated labels \cite{papernot2018scalable,wang2019private} from teacher models; another line of studies firstly sanitize the teacher model's label locally \cite{sun2020federated,lyu2020differentially}, and then aggregate these local $\epsilon$-DP labels. Recently, for bridging the advantages of centralized DP (i.e., high accuracy) and local DP (i.e., minimum trust), several studies \cite{liu2020flame,feldman2021hiding} propose shuffling messages from clients.

Note that all these studies consider client-level DP preservation, because of the impossibility to bound the worst-case contribution of one record in the teacher model. The most closely related works \cite{zhu2020private,zhu2020voting} try to account for one single record's contribution when distilling knowledge with $k$-NN. However, they only guarantee an approximate\&data-dependent version of DP. In contrast, our work realizes rigorous pure DP accountant at the record level and achieves significant accuracy improvements.

%% file: preliminary.tex
\section{Preliminaries}\label{sec:preliminary}

\subsection{Federated Knowledge Distillation}
Every data record $(x,y)$ is sampled from a Cartesian domain $\mathcal{X}\times\mathcal{Y}$, where the sample $x$ might be a tabular vector, an image, and etc. The class label $y$ can be a binary value (i.e. $\mathcal{Y}=\{0,1\}$) or categorical value (e.g., $\mathcal{Y}=\{0,1,..,9\}$ in hand-written digit recognition). Assuming that the client $i$ possesses $m_i$ records, we let $D^i$ denote these records as $D^i=[(x^i_{1},y^i_{1}),...,(x^i_{m_i},y^i_{m_i})]$, and let $D_{priv}$ denote the union of all local datasets:
$$D_{priv}=\bigcup_{i=1}^{n} D^i.$$

Let $D_{pub}=[(x_1,?),...,(x_{m_{pub}},?)]$ denote the public unlabeled dataset possessed by the federation orchestrator, the primary goal of federated knowledge distillation/transfer is then labeling $D_{pub}$ with knowledge from the $D_{priv}$. Current approaches \cite{papernot2018scalable,li2019fedmd,lyu2020differentially} are utilizing proxy models (e.g., neural networks) for answering labeling queries (every model is locally trained on $D^i$).

\subsection{Centralized Differential Privacy}
For datasets $D_{priv}$, $D'_{priv}$ that are of the same size and differ only in \emph{one record}, they are called \emph{neighboring datasets}. The centralized DP \cite{dwork2006differential}  at record-level with budget $(\epsilon,\delta)$ is as follows. As comparison, the client-level DP corresponds to neighboring datasets differ at one client's data.

\begin{definition}[Centralized $(\epsilon,\delta)$-DP]\label{def:cdp}
Let $\mathcal{K}$ denote the output domain, a randomized mechanism $K$ satisfies $\epsilon$-differential privacy iff for any neighboring datasets $D, D'$, and any outputs $\mathbf{z} \subseteq \mathcal{K}$,
$$\mathbb{P}[K(D) \in \mathbf{z}]\leq \exp(\epsilon)\cdot \mathbb{P}[K(D')\in \mathbf{z}]+\delta.$$
\end{definition}

\subsection{Local Differential Privacy}
Let $K$ denote a randomized mechanism for sanitizing a single record, the local DP \cite{duchi2013local} is as follows. 

\begin{definition}[Local $\epsilon$-DP]\label{def:ldp}
Assuming each client holds a dataset $D$ with only one record, a randomized mechanism $K$ satisfies local $\epsilon$-differential privacy iff for any data pair $D,D' \in \mathcal{X}\times \mathcal{Y}$, and any output $z \in \mathcal{K}$,
$$\mathbb{P}[K(D) = z]\leq \exp(\epsilon)\cdot \mathbb{P}[K(D')=z].$$
\end{definition}

\subsection{Shuffle Differential Privacy}
As a remedy to the low utility issue of the local privacy model, researchers \cite{bittau2017prochlo} propose the shuffle privacy model, where semi-trusted shufflers (or anonymous channels) lie between clients and the server. Let $K$ denote the local randomizer, and $t^i=K(D^i)$ denote the private message(s) from the client $i$, the definition of shuffle $(\epsilon,\delta)$-DP is as follows.

\begin{definition}[Shuffle $(\epsilon,\delta)$-DP]\label{def:sdp}
The randomized mechanism $K$ satisfies shuffle $(\epsilon,\delta)$-differential privacy iff the unordered union set $\bigcup_{i=1}^{n} t^i$ satisfies centralized $(\epsilon,\delta)$-DP constraints for the dataset $D_{priv}$.
\end{definition}

%% file: nonprivate.tex
\section{Reverse $k$-NN Labeling for KD}\label{sec:nonprivate}
We now introduce the reverse $k$-NN labeling method for federated knowledge distillation. To clarify its principles of design, we compare it with the conventional $k$-NN. Without loss of generality, we here consider image classification and it can be applied to other domain (e.g., tabular data, natural language, video, etc.) without much effort.

\begin{algorithm}[tb]
\caption{Record-level private KD}
\label{alg:rknn}
\textbf{Input}: $n$ clients, private datasets $\mathcal{D}^1,...,\mathcal{D}^n$, unlabeled public dataset $\mathcal{D}_{pub}$.\\
\textbf{Parameter}:  number of iterations $T$, number of nearest neighbors $k$, number of query samples $s$, privacy budget $\epsilon$.\\
\textbf{Output}: student model $M_S$ satisfies $\epsilon$-differential privacy.
\begin{algorithmic}[1]
\FOR{t=1,2,...,T}
    \STATE{select $s$ query samples $\mathcal{D}_{query}$ from $\mathcal{D}_{pub}$}
    \STATE{\color{Azure4}{// Client side}}
    \FOR{i=1,2,...,n}
        \STATE{connect each local record $(x_j^i, y_j^i) \in \mathcal{D}^{i}$ to $k$-nearest queries in $\mathcal{D}_{query}$}
        \STATE{find $k$-nearest neighbors $N_j^i \subseteq[1:s]$ of $x_j^i$}
        
        \STATE{represent the labeling answer on each query $l\in [1:s]$ as $a^i_l=\sum \llbracket l\in N^i_j \rrbracket y^i_j \in \mathbb{R}^{|\mathcal{Y}|}$}
        \STATE{send all labeling answers $A^i=[a^i_1,a^i_2,...,a^i_s]$}
    \ENDFOR
    \STATE{\color{Azure4}{// Sever side}}
    \STATE{aggregate the label counts $a_l=\sum\limits_{i=1}^{n}a_l^i \in \mathbb{R}^{|\mathcal{Y}|}$ from all clients for each $l\in [1:s]$}
    \STATE{ensemble $A=[a_1,a_2,...,a_s]$ and add noise to $A$ by $\epsilon$-differential privacy} 
    \STATE{derive labels $\{\hat{y}_l\}_{l=1}^s$ from noisy counts $A$}
    \STATE{train student model $M_S$ on $\mathcal{D}_{query}$ with labels $\{\hat{y}_l\}_{l=1}^s$}
\ENDFOR
\STATE \textbf{return} $M_S$
\end{algorithmic}
\end{algorithm}

\subsection{Methodology}
In the reverse $k$-nearest-neighbors labeling, for limiting every private record's contribution, each record is associated with at most $k$ (nearest) query samples; in order to improve the labeling accuracy, learned representations (instead of raw pixels) are utilized for distance measurement. At each learning iteration (see Algorithm \ref{alg:rknn}), the method follows four steps.

\paragraph{Learning to Represent:} Since raw pixels are unstable w.r.t. semantic labels, we measure sample distance by their latent representations. One can use pre-trained representation models (e.g., vision Transformers), or train an unsupervised representation model from scratch (e.g., via self-supervised learning \cite{chen2020simple}) with public-available $D_{pub}$.

\paragraph{Selecting Queries:} Labeling all samples in $D_{pub}$ is communication/computation/privacy expensive. Follow current approaches \cite{papernot2018scalable,wang2019private}, we select representative samples from $D_{pub}$. At the first iteration, we cluster $D_{pub}$ into $s$ groups in the representation space, and treat cluster centers $Q=[q_1,q_2,...,q_s]$ as query samples. For later iterations, samples are selected adaptively w.r.t. uncertainty of the current model $M_S$.   

\paragraph{Local Labeling:} Given queries $Q=[q_1,q_2,...,q_s]$, the client $i$ connects these query samples with local records under the reverse $k$-NN rule. Every local record $(x^i_j,y^i_j)$ is connected to $k$ nearest query samples (in the representation space). Assume that the label $y^i_j$ is presented in the one-hot vector form, and let $N^i_j\subseteq [1:s]$ denote the set of query indices that are $k$ nearest neighbors of $x^i_j$, then the labeling answer from client $i$ is $A^i=[a^i_1,a^i_2,...,a^i_s]$, where $a^i_l=\sum_{j=1}^{m_i}\llbracket l\in N^i_j\rrbracket y^i_j.$

\paragraph{Label Aggregation:} Given the labeling answers $A^1,...,A^n$ from clients, we summarize them as $A=[a_1,a_2,...,a_s]$, where
$a_l=\sum_{i=1}^m a^i_l$.
The final hard labeling results is then $(q_1,\hat{y}_1),...,(q_s,\hat{y}_s)$, where for $l\in [1:s]$:
$$\hat{y}_l=\arg \max_{c=1}^{|\mathcal{Y}|} a_l(c).$$
The final soft labeling results are $(q_1,\overline{y}_1),...,(q_s,\overline{y}_s)$, where
$\overline{y}_l=\frac{a_l}{\sum_{c=1}^{|\mathcal{Y}|} a_l(c)}$.
After assigning every sample in $D_{pub}$ with the label of its cluster center, a student model is built upon labeled  $\mathcal{D}_{pub}$ (iteration $1$) or $\mathcal{D}_{query}$ (iterations $[2:T]$) with conventional cross-entropy loss.

The proposed method is highly efficient, the computational/communication cost of the client $i$ is linear to number of local samples $|D^i|$ and queries $s$ (i.e., $O(|D^i|\cdot s\cdot T + s\cdot |\mathcal{Y}|\cdot T)$ and $O(s\cdot |\mathcal{Y}|\cdot T)$). Compared to the classical knowledge distillation paradigm, it has the advantage of avoiding local model training, thus fits both resource-rich cross-silo and resource-scarce cross-device settings.

Note that the summarized labeling answer $A$ is independent from how private samples $D_{priv}$ distribute among clients, and only depends on the whole $D_{priv}$. Therefore, the above method resists non-IID settings.

\subsection{$k$-NN vs. Reverse $k$-NN}\label{subsec:vs}
To tell apart our record-level approach from the current model-level paradigm for federated knowledge distillation, we here compare the reverse $k$-NN with $k$-NN labeling. 

When the $k$-NN classifier works as a local model for labeling (e.g., in \cite{zhu2020private,zhu2020voting}), each query sample is associated with at most $k$ records. However, from a record's perspective, it might be associated with all $s$ query samples, hence its maximum-possible contribution (i.e., the sensitivity in differential privacy) to the final answer $A$ is $\Theta(s)$. As a comparison, one record's contribution in the reverse $k$-NN is bounded by $\Theta(k)$ and is much smaller than $\Theta(s)$. This difference in worst-case contribution causes dramatic gaps when seeking a privacy/utility trade-off.

Note that seeking other connection rules having the same privacy guarantee as reverse $k$-NN is also possible, please refer to Appendix C for detail.  

%% file: central.tex
\section{Centralized Private Mechanisms}\label{sec:central}
In this section, we reformulate the reverse $k$-NN labeling as the problem of Bucketized Sparse Vector Summation (BSVS), then present centralized DP mechanisms for the problem, and provide corresponding labeling error bounds.


\subsection{Reformulation}
The key steps in the reverse $k$-NN labeling can be abstracted as Bucketized Sparse Vector Summation (see Definition \ref{def:bsvs}). Compared to the histogram summation and generalized bucketed vector summation \cite{chang2021locally}, a critical difference is that the vector we consider is sparse, as the label $y^i_j$ is one-hot (in multi-class classification) or multi-hot (in multi-label classification).

\begin{definition}[Bucketized Sparse Vector Summation]\label{def:bsvs}
In the BSVS problem, each datum corresponds to a set $T_j \subseteq T$ of $k$ buckets and a sparse vector $y_j \in \{0,1\}^{|\mathcal{Y}|}$ and $|y_j|=r$ . The goal is to determine, for a given $t\in T$, the vector sum of bucket $t$, which is $a_t := \sum_{j=1}^{\sum_{i=[1:n]}m_i}\  y_j \llbracket t\in T_j \rrbracket$. An approximate oracle $\tilde{a}$ is said to be $(\eta,\beta)$-accurate at $t$ if we have $|a_t-\tilde{a}_t|_{+\infty}<\eta$ with probability $1-\beta$.
\end{definition}


In the above reformulation, the number of buckets is equal to the number of query samples: $|T|=s$. Note that in conventional multi-class classification, we have $r\equiv 1$.


\subsection{Mechanism and Accuracy Guarantees}
When centralized $(\epsilon,0)$-DP is imposed on the BSVS problem, we employ the classical Laplace mechanism for privacy preservation. Apparently, the sensitivity $\Delta$ is the maximum possible magnitude of $A^i=[a^i_1,a^i_2,...,a^i_s]$ (i.e., $2k\cdot r$). Therefore, we inject $Laplace(\frac{2k\cdot r}{\epsilon})$ to every element of $A$. The corresponding accuracy guarantee is presented in Proposition \ref{pro:laplace}, which is derived from the tail probability bound of the Laplace distribution.

\begin{proposition}\label{pro:laplace}
There is an $(\frac{2k\cdot r \cdot \log(|\mathcal{Y}|/\beta)}{\epsilon}, \beta)$-accurate centralized $\epsilon$-DP algorithm for the BSVS problem.
\end{proposition}

For the $t$-th query/bucket, define the (non-private) count gap between the true label $y^*\in[1:|\mathcal{Y}|]$ and false labels as: 
$$Gap_t=a_t(y^*)-\max_{c\in [1:|\mathcal{Y}|]\ and\ c \neq y^*} a_t(c),$$
we then have the following conclusion on the private labeling accuracy w.r.t. the accuracy of the BSVS problem:
\begin{remark}
If $Gap_t\geq 2\alpha$ and the private algorithm is $(\alpha, \beta)$-accurate, then with probability $1-\beta$, the estimated hard labeling result is accurate (equals to the true label $y^*$).
\end{remark}

%% file: local.tex
\section{Local Private Mechanisms}\label{sec:local}
Considering the most stringent case of imposing local DP on every client who holds only one record (i.e., $m_i\equiv 1$), every client $i$ now sanitize the labeling answer $A^i=[a^i_1,a^i_2,...,a^i_s]$ independently. Note that this case also fits cases one client holds multiple records, if we sample one record or simply normalizing labeling answers $A^i$.
Naively, we could also adopt the Laplace mechanism and add $Laplace(\frac{2\cdot k \cdot r}{\epsilon})$ to every element in $A^i$. However, it is dominated by the randomized response mechanism \cite{duchi2013local}, which randomly flips every binary value in $A^i$ with probability $\frac{1}{e^{\epsilon/(2k r)}+1}$. We show randomized response is $(O(\sqrt{\frac{n k^2 r^2\log(|\mathcal{Y}|/\beta)}{ \epsilon^2}}),\beta)$-accurate (in Theorem \ref{pro:rr}).
\begin{theorem}\label{pro:rr}
The local $\epsilon$-DP randomized response mechanism is an $(\frac{e^{\epsilon/(2k r)}+1}{e^{\epsilon/(2k r)}-1}\sqrt{{3n\log(|\mathcal{Y}|/\beta)/(e^{\epsilon/(2k r)}+1)}}, \beta)$-accurate algorithm for the BSVS problem when $\epsilon=O(1)$. 
\end{theorem}
\begin{proof}
Recall that for a binary value $b$ flipped with probability $\frac{1}{e^{\epsilon/(2k r)}+1}$, the unbiased estimation given the observation $b'$ is $\tilde{b}=\frac{b'-{1}/(e^{\epsilon/(2k r)}+1)}{(e^{\epsilon/(2k r)}-1)/(e^{\epsilon/(2k r)}+1)}$. The total count of observed ones is a summation of $n$ Bernoulli variables with a success rate of either $\frac{1}{e^{\epsilon/(2k r)}+1}$ or $\frac{e^{\epsilon/(2k r)}}{e^{\epsilon/(2k r)}+1}$. Let $u$ denote the estimation bias of one element in $\tilde{a}_t$, we have $\mathbb{P}[|u|>\eta \cdot \frac{e^{\epsilon/(2k r)}+1}{e^{\epsilon/(2k r)}-1}]\leq \exp(\frac{-\eta^2 (e^{\epsilon/(2k r)}+1)}{3n})$. Therefore, with probability of $1-\beta$, we have $|a_t-\tilde{a}_t|_{+\infty}\leq \frac{e^{\epsilon/(2k r)}+1}{e^{\epsilon/(2k r)}-1}\sqrt{\frac{3n\log(|\mathcal{Y}|/\beta)}{e^{\epsilon/(2k r)}+1}}$.
\end{proof}

Due to budget splitting, the randomized response is suffering from the error rate of $\tilde{\Theta}(\frac{k\cdot r}{\epsilon})$. We can actually adopt an optimal sparse vector summation oracle (in the high privacy regime) \cite{wanghiding} for the BSVS problem and achieve an error rate of $\tilde{\Theta}(\frac{\sqrt{k\cdot r}}{\epsilon})$ (see Appendix A and B).

%% file: shuffle.tex
\section{Shuffle Private Mechanisms}\label{sec:shuffle}
When messages from users are anonymized \& shuffled by anonymous channels or shufflers, the server only observes a multi-set about messages. Consequently, to achieve a certain level of (centralized) differential privacy, every client can inject fewer noises in the local. According to the number of messages one client may publish, the shuffle privacy model can be categorized into the multi-message one \cite{ghazi2020private} and the single-message one \cite{feldman2021hiding}.   


\subsection{Multi-message Shuffling}
In the multi-message shuffle privacy model, the basic idea is to add noises to $A$ in a distributed manner. For the categorical distribution estimation problem with dimension $d$, \cite{ghazi2020private} proposes an $(\epsilon,\delta)$-DP protocol with an error equal to adding independent $Laplace(\frac{4}{\epsilon})$, and with expected messages of one user equal to $1+O(\frac{d \log^2(1/\delta)}{\epsilon^2 n})$, each consisting $\lceil \log d \rceil +1$ bits. In the protocol, the $Laplace(\frac{4}{\epsilon})$ is decomposed into $\Theta(n)$ negative binomial variables and added to every entry in $A^i$ accordingly. Each message is an index in $[1:d]$, which means plus one to the index. For the BSVS problem with $(\epsilon,\delta)$-DP, follow almost the same protocol in \cite{ghazi2020private}, the $Laplace(\frac{4k r}{\epsilon})$ can be added in a distributed manner with expected messages of one client equal to $k r+O(\frac{d k^2 r^2 \log^2(1/\delta)}{\epsilon^2 n})$, each consisting $\lceil \log d \rceil +1$ bits. Plugging into the analyses on Laplace mechanism in Proposition \ref{pro:laplace}, we conclude that it is $(\frac{4 k r \log(|\mathcal{Y}|/\beta)}{\epsilon},\beta)$-accurate.

\subsection{Single-message Shuffling}
When each client is constrained to send only one message to the shuffler, the message must be local DP \cite{cheu2019distributed}, while the privacy in the central perspective is amplified. Recently, \cite{feldman2021hiding} gives a tight privacy amplification bound for any local private mechanisms, shows $n$ local $\epsilon$-DP and shuffled messages satisfy centralized $(\log(1+(\frac{8\sqrt{e^\epsilon \log(4/\delta)}}{\sqrt{n}}+\frac{8 e^\epsilon}{n})\frac{e^\epsilon-1}{e^\epsilon+1}),\delta)$-DP when $\log(\frac{n}{16 \log(2/\delta)})\geq \epsilon$. In return, when centralized privacy budget $(\epsilon,\delta)$ is given, we can reversely derive the enlarged local budget, and provide corresponding accuracy guarantees with Theorem \ref{pro:rr}.





%% file: experiment.tex
\section{Experiments}\label{sec:experiment}
To validate proposed record-level private mechanisms for federated learning, we conduct extensive experiments on real-world datasets  to answer the following questions: \textbf{(1)} What is the effect of query selection on accuracy? \textbf{(2)} What is the effect of parameter $k$ on accuracy? \textbf{(3)} What is the performance gap between our approach and SOTA methods?

The competitive approaches include SOTA private knowledge distillation methods by adding Laplace noise (\emph{LNMAX}) or Gaussian noise (\emph{GNMAX}) in \cite{papernot2016semi,papernot2018scalable}, private $k$-NN \cite{zhu2020private} and the noisy SGD methods in \cite{luo2021scalable}. Our approach is implemented with $(\epsilon,0)$-DP, while competitive approaches are implemented with $(\epsilon,10^{-5})$-DP. Notice that given the representation model trained on the $D_{pub}$ and the labeled samples in $D_{priv}$, one may also simply train a prediction head with noisy SGD \cite{abadi2016deep} as the classifier. We denote this straight-forward approach as \emph{Layer-$1$ noisy SGD} (with one prediction layer) and \emph{Layer-$2$ noisy SGD} (with two prediction layers). Both Layer-$1$ and Layer-$2$ noisy SGD train parameters with the size of few thousands, which is much smaller than the whole classifier's, but is still much larger than the label domain size $|\mathcal{Y}|$.


\subsection{Datasets and Networks}
Three popular image datasets are employed for experiments: \textbf{MNIST} \footnote{http: //yann.lecun.com/exdb/mnist} that contains $70,000$ gray-scale images of size $28\times 28$, has $10$ categories;
\textbf{SVHN} \footnote{http://ufldl.stanford.edu/housenumbers} that contains $630,420$ digit images of size $32\times 32$ and $10$ categories;
\textbf{CIFAR-10} \footnote{https://www.cs.toronto.edu/~kriz/cifar.html} that contains $60,000$ images of size $32\times 32$ and $10$ categories.


Following common settings in the literature, for the MNIST, we assume the public data $D_{pub}$ is $5,000$ samples from the test dataset, the remaining $5,000$ test samples are used for evaluating the performance of the student classifier, and the training dataset is used as the private data $D_{priv}$;
for the SVHN, we assume the public data $D_{pub}$ is $26k$ samples from the test dataset, the remaining $1k$ test samples are used for evaluating the performance of the student classifier, and the training dataset, together with the extended data, is used as the private data $D_{priv}$;
for the CIFAR-10, we assume the public data $D_{pub}$ is $30,000$ samples from the training set, the $1,000$ samples from the test dataset are used for evaluation, and use other $29,000$ samples as the private data $D_{priv}$. The experimental results of shuffle DP is omitted, since it is analogy to centralized or local DP.


For the MNIST dataset, the architecture of the student classifier is from \cite{an2020ensemble}, and the DTI \cite{monnier2020deep} is employed for general purpose representation \& clustering on $D_{pub}$ (denoted as \textbf{[general]}).
For the SVHN dataset, the architecture of the student classifier is Mixmatch \cite{berthelot2019mixmatch}, and the histogram of oriented gradients(HOG) \cite{dalal2005histograms} and $k$-means++ is employed for general purpose representation \& clustering on the $D_{pub}$.
For the CIFAR-10 dataset, the network architecture is DenseNet121, and the SimCLR \cite{chen2020simple} and $k$-means++ is used for representation learning \& clustering on the $D_{pub}$.



\subsection{Performance Metrics}
Two accuracy indications are employed for measuring the performances, one is the accuracy of the private label answering ($\text{Acc}_{pl}$), the other is the test accuracy of the privately learned classifier ($\text{Acc}_{pc}$). As we use unsupervised clustering for query selection at iteration $1$, here the $\text{Acc}_{pl}$ is the number of public samples receiving correct labels divided by $|D_{pub}|$.


\subsection{Varying Number of Clusters}
The purity of clusters (w.r.t. class labels) upper bounds the $\text{Acc}_{pl}$. Increasing $s$ can roughly increase purity, but reduce the number of local records associated with one query. We here explore the appropriate number $s$. When $s$ When s changes within a certain range, we present the experimental results $\text{Acc}_{pc}$ in Figure \ref{fig:mnistsk} for MNIST (with $\epsilon=0.1$),  Figure \ref{fig:svhnsk} for SVHN (with $\epsilon=0.1$), and Figure \ref{fig:cifar10sk} for CIFAR-10 (with $\epsilon=1$). For the more simple MNIST, the best number of clusters is around $40$; while for the SVHN/CIFAR-10, the accuracy increases with $s$, since it has more diversity in one class. The (omitted) experimental results on labeling accuracy $\text{Acc}_{pl}$ are always $0.5\%$-$3.0\%$ behind the $\text{Acc}_{pc}$, imply building the classifier could suppress labeling noises due to privacy preservation.

\begin{figure}[h]
\vspace*{-0.7em}
\begin{center}
\centerline{\includegraphics[width=80mm]{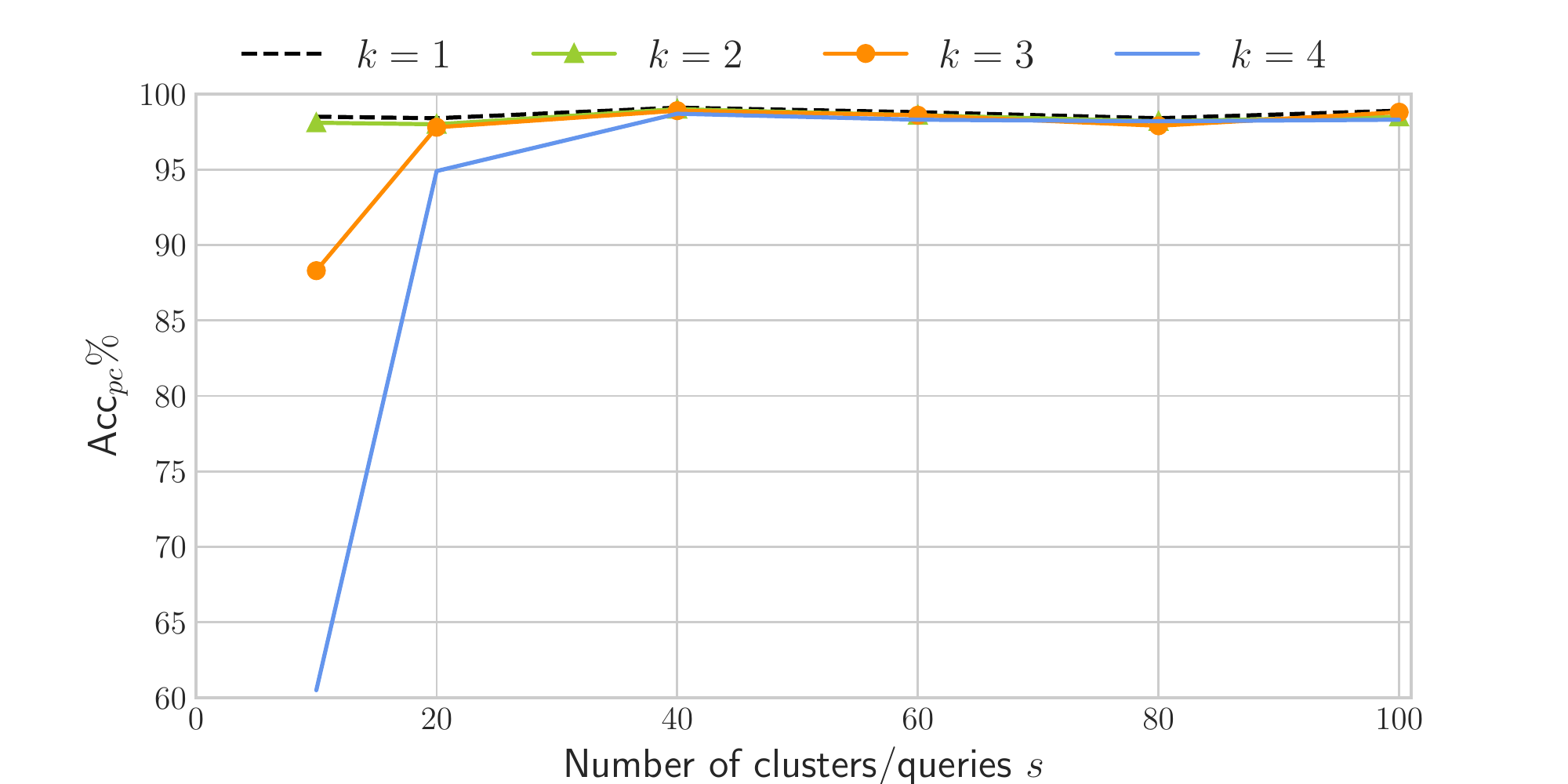}}
\vskip -0.1in
\caption{Experimental results on MNIST ($\epsilon=0.1$, $T=1$) with vary number of clusters $s$ and vary number of neighbors $k$.}
\label{fig:mnistsk}
\end{center}
\vspace*{-1.7em}
\end{figure}

\begin{figure}[h]
\vspace*{-0.7em}
\begin{center}
\centerline{\includegraphics[width=80mm]{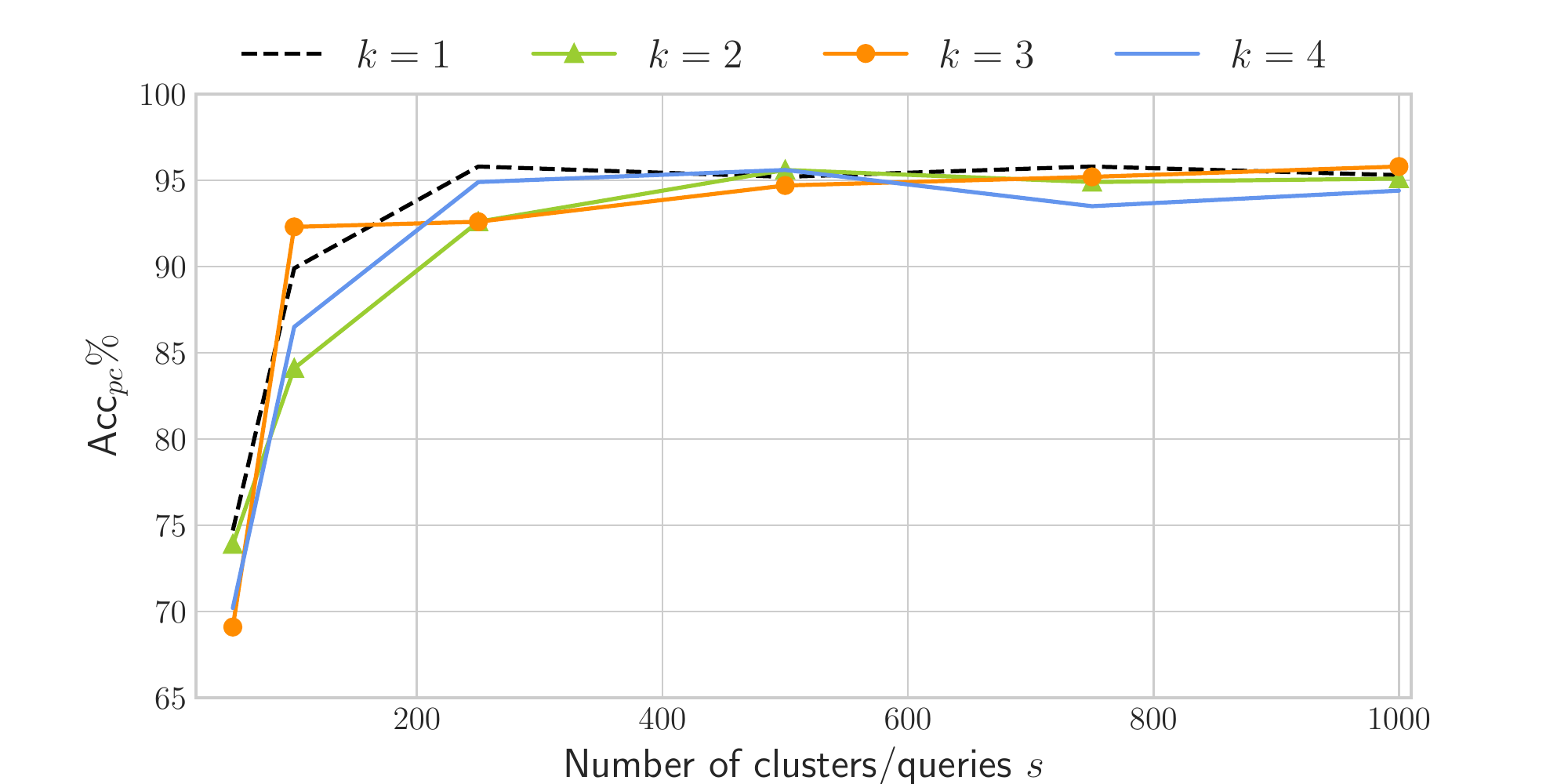}}
\vskip -0.1in
\caption{Experimental results on SVHN ($\epsilon=1$, $T=1$) with vary number of clusters $s$ and vary number of neighbors $k$.}
\label{fig:svhnsk}
\end{center}
\vspace*{-1.7em}
\end{figure}

\begin{figure}[h]
\vspace*{-0.7em}
\begin{center}
\centerline{\includegraphics[width=80mm]{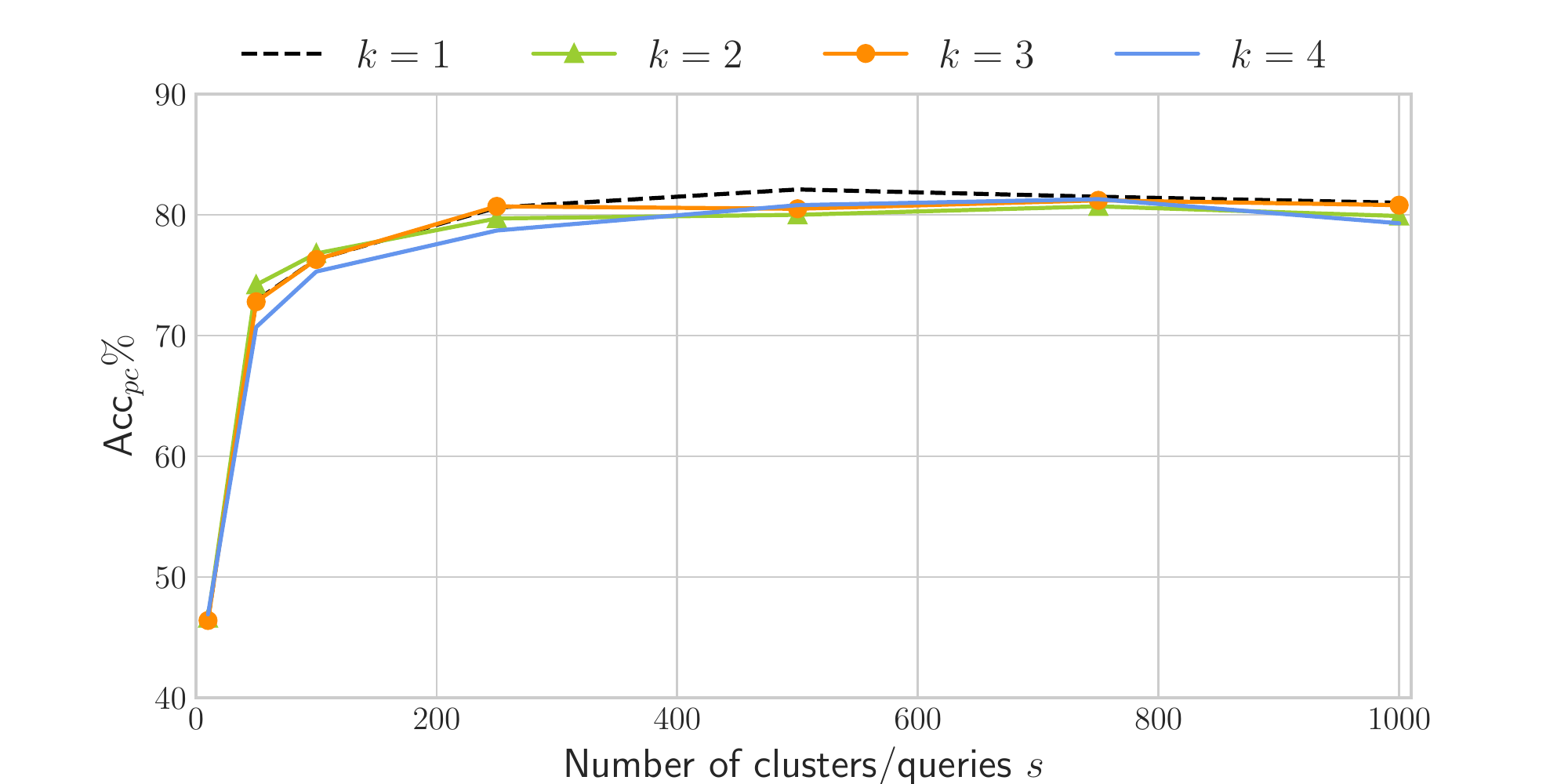}}
\vskip -0.1in
\caption{Experimental results on CIFAR-10 ($\epsilon=1$, $T=1$) with vary number of clusters $s$ and vary number of neighbors $k$.}
\label{fig:cifar10sk}
\end{center}
\vspace*{-1.7em}
\end{figure}

\begin{table*}
\renewcommand{\arraystretch}{1.06}
\caption{Test accuracy \& privacy consumption comparison of centralized differentially private methods.}
\vskip -0.1in
\label{tab:comparison}
\centering
\begin{tabular}{c|cccccc}
\hline\bfseries Dataset & \bfseries Methods &\bfseries $\#$Queries &\bfseries $\epsilon$ &\bfseries Test Acc. &\bfseries Label Acc. &\bfseries Non-priv Acc. \\
\hline
\multirow{8}{*}{MNIST} & LNMAX \cite{papernot2018scalable} & $1000$ & $8.03$ & $98.1\%$ & & \multirow{3}{*}{$99.2\%$} \\
 & GNMAX \cite{papernot2018scalable} & $286$ & $1.97$ & $98.5\%$ & & \\
 & Private $k$-NN \cite{zhu2020private}  & $735$ & $0.47$ & $98.8\%$ &  & \\
 & Noisy SGD \cite{abadi2016deep}  &  & $1.0$ & $81.2\%$ & &$91.1\%$ \\
 \cdashline{2-7}
 & Ours \textbf{[general]} & $40$ & ${\mathbf{0.1}}$ & ${\mathbf{99.1\%}}$ & ${98.5\%}$ & \multirow{2}{*}{$99.2\%$} \\
 & Ours \textbf{[general]} & $40$ & ${\mathbf{0.04}}$ & ${\mathbf{98.6\%}}$ & ${97.7\%}$ & \\\cdashline{2-7}
 & Ours \textbf{[end2end]} & $10$ & ${\mathbf{0.01}}$ & ${\mathbf{98.5\%}}$ & ${97.5\%}$ & \multirow{2}{*}{$98.7\%$}  \\
 & Ours \textbf{[end2end]} & $10$ & ${\mathbf{0.004}}$ & ${\mathbf{98.2\%}}$ & ${97.3\%}$ & \\
  
\hline
\multirow{6}{*}{SVHN} &  LNMAX \cite{papernot2018scalable} &
$1000$ & $8.19$ & $90.1\%$ & & \multirow{3}{*}{$92.8\%$} \\
& GNMAX \cite{papernot2018scalable} & $3098$ & $4.96$ & $91.6\%$ & & \\
   & Private $k$-NN \cite{zhu2020private}  & $2939$ & $0.49$ & $91.6\%$ &  & \\
 & Noisy SGD \cite{abadi2016deep}  &  & $4.0$ & $76.0\%$ & &$84.4\%$ \\
 \cdashline{2-7} 
& Ours \textbf{[general]} & $500$ & ${\mathbf{0.1}}$ & ${\mathbf{95.6\%}}$ &  & \multirow{2}{*}{$96.7\%$} \\
 & Ours \textbf{[general]} & $500$ & ${\mathbf{0.04}}$ & ${\mathbf{95.3\%}}$ &  & \\
\hline
\multirow{7}{*}{CIFAR-10}   & GNMAX \cite{papernot2018scalable} & $286$ &  & $<50\%$ &  & \multirow{2}{*}{$80.5\%$} \\
& Private $k$-NN \cite{zhu2020private}  & $3877$ & $2.92$ & $70.8\%$ \\
& Finetuning Noisy SGD  \cite{luo2021scalable}  &  & $1.0$ & $76.5\%$ & &  \\
& Layer-$1$ Noisy SGD \cite{abadi2016deep}  &  & $4.0$ & $73.7\%$ & & $77.7\%$  \\
& Layer-$2$ Noisy SGD \cite{abadi2016deep}  &  & $4.0$ & $78.5\%$ & & $80.9\%$  \\

\cdashline{2-7}
& Ours \textbf{[general]} & $500$ & $\mathbf{1.0}$ & $\mathbf{82.1\%}$ & $77.1\%$ & \multirow{2}{*}{$82.3\%$} \\
& Ours \textbf{[general]} & $500$ & $\mathbf{0.29}$ & $\mathbf{79.4\%}$ & $73.9\%$ & \\
\cdashline{2-7}
& Ours \textbf{[end2end]} & $10$ & $\mathbf{0.01}$ & $\mathbf{86.1\%}$ & $85.9\%$ & \multirow{2}{*}{$86.2\%$} \\
& Ours \textbf{[end2end]} & $10$ & $\mathbf{0.005}$ & $\mathbf{86.0\%}$ & $85.7\%$  \\
\hline
\end{tabular}
\end{table*}

\subsection{Varying $k$ in Reverse Nearest Neighbors}
We also explore the choice of $k$ of reverse $k$-NN in Figures \ref{fig:mnistsk}, \ref{fig:svhnsk}, and \ref{fig:cifar10sk}. It is demonstrated that there is no noticeable difference between choosing $k$ at $1$, $2$, $3$ or $4$. Theoretically, as $k$ gets larger, the label count of each query grows with $k$, but the count gap grows sublinear to $k$ and the standard devivation of the privacy noise grows with $k$. Here the $k$ in $[1,2,3,4]$ are all small, thus the sublinerity is negligible and the noises hardly overwhelm count gaps. It is experimentally observed that when $k$ grows to about $10$, the performances begin to drop significantly.




\subsection{Local DP}
For the most stringent case that local DP is imposed on every client's single record, we present results in Figure \ref{fig:ldpep} for MNIST/CIFAR-10. Since noises due to local DP easily dominate $Gap_t$, we here fix hyper-parameters at $s=|\mathcal{Y}|=10$ and $k=1$, and employ end-to-end unsupervised clustering on MNIST with DTI \cite{monnier2020deep} and CIFAR-10 with SCAN \cite{van2020scan}. It is observed that the Collision mechanism \cite{wanghiding} achieves test accuracy of $98.5\%$ for MNIST and $78.2\%$ for CIFAR-10 with privacy budget $\epsilon=0.4$. To the best of our knowledge, it is the first time local private deep learning provides meaningful privacy/accuracy trade-offs. 


\begin{figure}[h]
\vspace*{-0.5em}
\begin{center}
\centerline{\includegraphics[width=80mm]{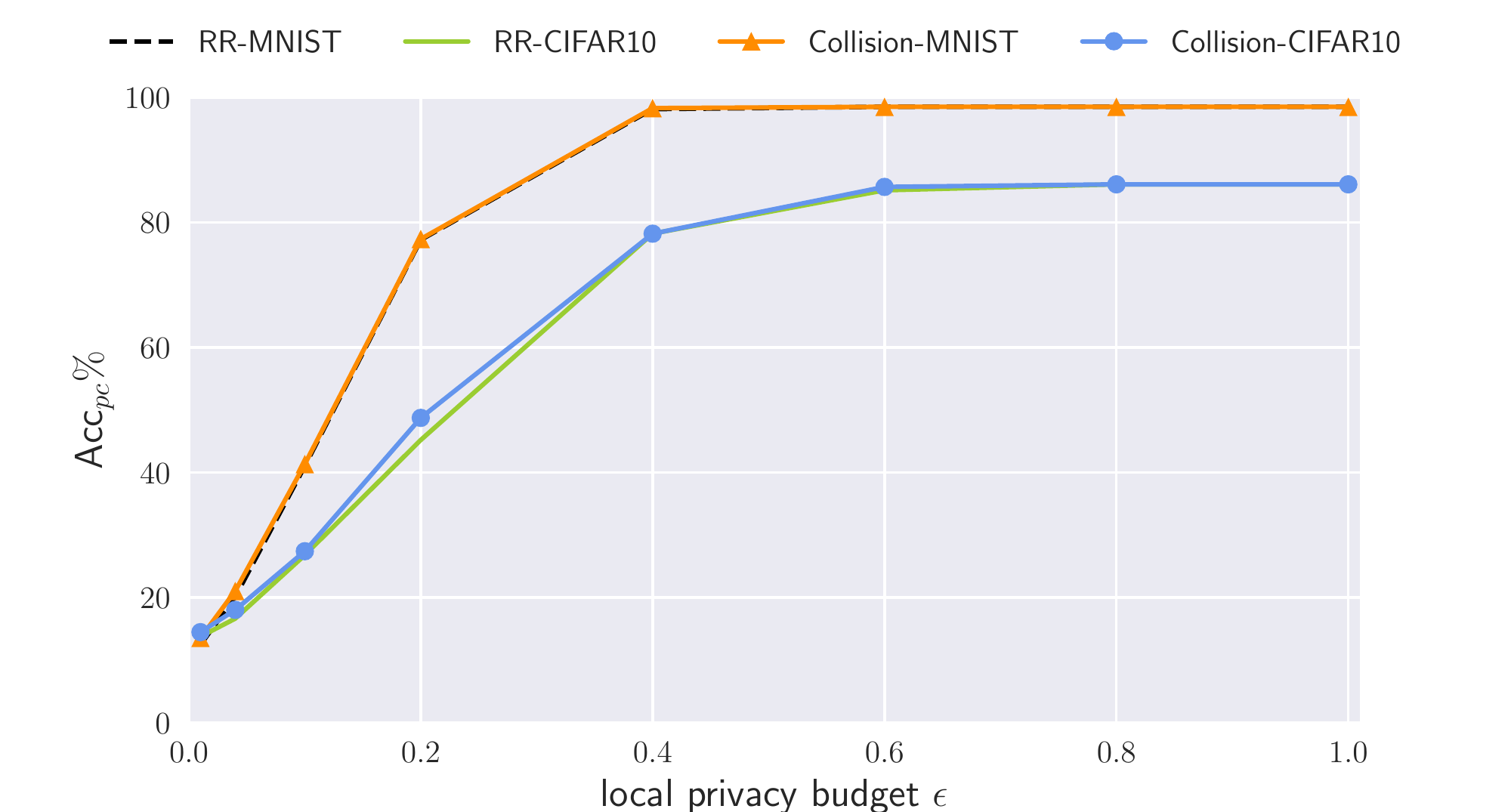}}
\vskip -0.1in
\caption{Local DP experimental results on MNIST/CIFAR-10 with randomized response (RR) and Collision mechanism (Collision), when $\epsilon$ ranges from $0.01$ to $1.0$.}
\label{fig:ldpep}
\end{center}
\vspace*{-1.5em}
\end{figure}


\subsection{Comparison with Existing Approaches}
In Table \ref{tab:comparison}, we compare our method' results with reported results of existing approaches with same settings. The hyper-parameter of our method is set to $T=1$ and $k=1$.
When utilizing general purpose unsupervised representation learning and clustering, compared to existing client-level DP methods (i.e., LNMAX, GNMAX) or record-level approximate DP methods (i.e., Noisy SGD, Private $k$-NN), our method achieves better accuracy with an order magnitude smaller privacy consumption. 
Specifically, if we employ end-to-end unsupervised clustering \cite{monnier2020deep,van2020scan} (denoted as \textbf{[end2end]}), we are able to achieve (average) accuracy of $86.1\%$ with centralized $\epsilon=0.01$ for CIFAR-10, and $99.1\%$ accuracy for MNIST. When the number of query $s=|\mathcal{Y}|$, we have the count gap $Gap_t\approx |D_{priv}|/s$ is tens of hundreds, which is large enough to overcome Laplace noises with standard deviation $2\sqrt{2}/\epsilon=2\sqrt{2}/0.01< 300$.

\subsection{Summary}
In summary, our record-level private knowledge distillation method is an effective way to centralized/decentralized machine learning, and significantly outperforms the Private $k$NN \cite{zhu2020private} that preserves only approximate\& data-dependent record-level privacy. The $82.1\%$ accuracy on CIFAR-10 also surpasses the SOTA accuracy $76.5\%$ (with $\epsilon=1$) of the noisy SGD method in \cite{luo2021scalable}, demonstrates the powerful privacy\&utility trade-off of knowledge distillation with record-level privacy. When equipped with tighter privacy accountant by R\'enyi differential privacy for our approach (in future study) or when data is non-I.I.D. across clients, the performance gaps can be even larger.  

%% file: conclusion.tex
\section{Conclusion}\label{sec:conclusion}
This work tackled one major drawback remaining in federated learning with knowledge distillation, and advocated for fine-grained record-level privacy preservation. We proposed the reverse $k$-NN labeling as a solution that limits every single record's contribution, and is naturally immue to non-I.I.D. settings. After formulating the reverse $k$-NN labeling as bucketized sparse vector summation (BSVS), we provided concrete differentially private mechanisms under comprehensive scenarios (i.e., in centralized/local/shuffle settings). Theoretically, these mechanisms are guaranteed for labeling accuracy, which is determined by privacy budget and label count gaps. Experimentally, our solution achieved $99.1\%$/$95.6\%$ test accuracy (with $\epsilon=0.1$) on the MNIST/SVHN dataset and $82.1\%$ test accuracy (with $\epsilon=1$) on the CIFAR-10 dataset, and improved significantly upon existing private knowledge-distillation/gradient-descent based methods with one magnitude lower of privacy consumption.

%% file: appendix.tex
\section*{APPDENDIX}

\subsection*{A. An Optimal Oracle for High Privacy}
Due to budget splitting, the randomized response is suffering from the error rate of $\tilde{\Theta}(\frac{k\cdot r}{\epsilon})$. In this part, we adopt an optimal sparse vector summation oracle (in the high privacy regime) \cite{wanghiding} for the BSVS problem and achieve an error rate of $\tilde{\Theta}(\frac{\sqrt{k\cdot r}}{\epsilon})$.

If we flatten the labeling answer $A^i$, essentially the BSVS problem is a special case of sparse vector summation where the domain size $d$ is $s\cdot |\mathcal{Y}|$ and the maximum cardinality $c$ is $k\cdot r$. One mean-squared-error optimal mechanism for sparse vector summation is the Collision mechanism \cite{wanghiding} (see Definition \ref{def:collision}), where all non-zero entries are mapped into a more dense Bloom filter with length $l$ via local hashes. 

\begin{definition}[$(d,c,\epsilon,l)$-Collision Mechanism \cite{wanghiding}]\label{def:collision}
Given a random-chosen hash function $H: \mathcal{G}\mapsto \mathcal{Z}$, take a vector $V$ having $c$ non-zero entries as the input ($V \subseteq \mathcal{G}$), the Collision mechanism randomly outputs an element $z \in \mathcal{Z}$ according to following probabilities:
\begin{equation}\label{eq:collision}
    \mathbb{P}[z | V]= \left\{
    \begin{array}{@{}lr@{}}
        \frac{e^\epsilon}{\Omega},\ \ \ \ \ \ \ \ \ \ \ \ \ \ \ \ \ \ \ \ \ \ \ \ \ \ \ \ \ \ \ \   \text{if } \exists v \in V, z=H(v);\\
        \frac{\Omega-e^\epsilon\cdot\#\{H(v)\ |\ H(v)\ for\ v\ \in V\}}{(l-\#\{H(v)\ |\ H(v)\ for\ v\ \in V\})\cdot \Omega}.\ \ \ \  otherwise.\\
    \end{array}
    \right.
\end{equation}
The normalization factor is $\Omega=c\cdot e^\epsilon+l-c$. An unbiased estimator of indicator $\llbracket v\in V\rrbracket$ (for $v \in \mathcal{G}$ when $c\geq 2$) is:
$$\tilde{\llbracket v \in V\rrbracket}=\frac{\llbracket H(v)= z\rrbracket-1/l}{{e^\epsilon}/{\Omega}-1/l}.$$  
\end{definition}

Setting the Bloom filter length $l$ at around $2k r-1+k r e^\epsilon$, we show that the Collision mechanism for the BSVS problem is $(O(\sqrt{\frac{n k r \log(|\mathcal{Y}|/\beta)}{\epsilon^2}}),\beta)$-accurate (see Theorem \ref{pro:collision}).
\begin{theorem}\label{pro:collision}
The $(s\cdot |\mathcal{Y}|,k\cdot r,\epsilon,2k\cdot r-1+k\cdot r\cdot e^\epsilon)$-Collision mechanism is an $(O(\sqrt{\frac{n k r \log(|\mathcal{Y}|/\beta)}{\epsilon^2}}), \beta)$-accurate algorithm for the BSVS problem when $\epsilon=O(1)$.
\end{theorem}
\begin{proof}
Recall that the binary value $\llbracket H(v)= z\rrbracket$ could be deemed as a Bernoulli variable with success rate of either $\frac{1}{l}$ or $\frac{e^\epsilon}{\Omega}$. Since the total count of observed ones is a summation of $n$ independent Bernoulli variables, let $u$ denote the estimation bias of one element in $\tilde{a}_t$, according to the multiplicative Chernoff bound and the fact that $\frac{e^\epsilon}{\Omega}\geq \frac{1}{l}$, we have $\mathbb{P}[|u|>\eta\cdot n\cdot \frac{1}{{e^\epsilon}/{\Omega}-1/l}]\leq \exp(\frac{-\eta^2 n}{2  l})$. Therefore, when $l=2k\cdot r-1+k\cdot r\cdot e^\epsilon$, with probability of $1-\beta$, we have $|a_t-\tilde{a}_t|_{+\infty}\leq \frac{(k\cdot r\cdot e^{\epsilon}+2k\cdot r-1)(2k\cdot r\cdot e^{\epsilon}+k\cdot r-1)}{k\cdot r\cdot(e^{2\epsilon}-1)-(e^{\epsilon}-1)} \sqrt{\frac{2n \log(|\mathcal{Y}|/\beta)}{l}}$. Applying $e^{\epsilon}\approx \epsilon+1$, we have the bound. 
\end{proof}


\subsection*{B. An Improved Oracle for Low Privacy}
In this part, we improve the accuracy of BSVS in the low privacy regime (e.g., when $\epsilon>1$ with shuffling privacy amplification in the next section). Note that the labeling answer $A^i$ of user $i$ equals the multiplication of $T_i \in [0,1]^{s\times 1}$ and $y_i \in [0,1]^{1\times |\mathcal{Y}|}$ (see detail in Definition \ref{def:bsvs}). Instead of treating $A^i$ as a flat vector, we could derive $A^i$ from privately estimated $T_i$ and $y_i$. There are two approaches to estimate  $T_i$ and $y_i$ simultaneously:


\begin{itemize}
    \item \textbf{Separation:\ } We evenly split the local budget into two parts, and estimate $T_i$ and $y_i$ separately. Using the sparse vector oracle of the Collision mechanism \cite{wanghiding}, the average mean squared error of each estimated entry $\hat{a}$ in $T_i$ (or $\hat{b}$ in $y_i$) is then approximately $O(\frac{k}{\epsilon^2})$ (or $O(\frac{r}{\epsilon^2})$). Since most entries in $T_i$ and $y_i$ are zero (i.e., $\mathbb{E}[\hat{a}]=\mathbb{E}[\hat{b}]=0$), for a multiplied entry $\hat{a}\cdot \hat{b}$ in $A^i$, we have its average mean squared error is approximately:
    $$\text{Var}[\hat{a}\cdot \hat{b}]=\text{Var}[a]\cdot \text{Var}[b]=O(\frac{k r}{\epsilon^4}).$$
    Compared to the approach in the previous subsection having error $O(\frac{k r}{\epsilon^2})$, when the privacy budget is relatively high, this approach is more accurate.
    \item \textbf{Concatenation:\ } We concatenate the $T_i$ and $y_i$ to compose a vector with length $s+|\mathcal{Y}|$ and $k+r$ non-zero entries. Utilizing the sparse vector oracle in previous subsections, we have the average mean squared error of each entry in $T_i$ or $y_i$ is then approximately $O(\frac{k+r}{\epsilon^2})$. Hence, for every multiplied entry in $A^i$, we have its average mean squared error is approximately $O(\frac{(k+r)^2}{\epsilon^4})$.
\end{itemize}

\begin{figure}
\begin{center}
\centerline{\includegraphics[width=81mm]{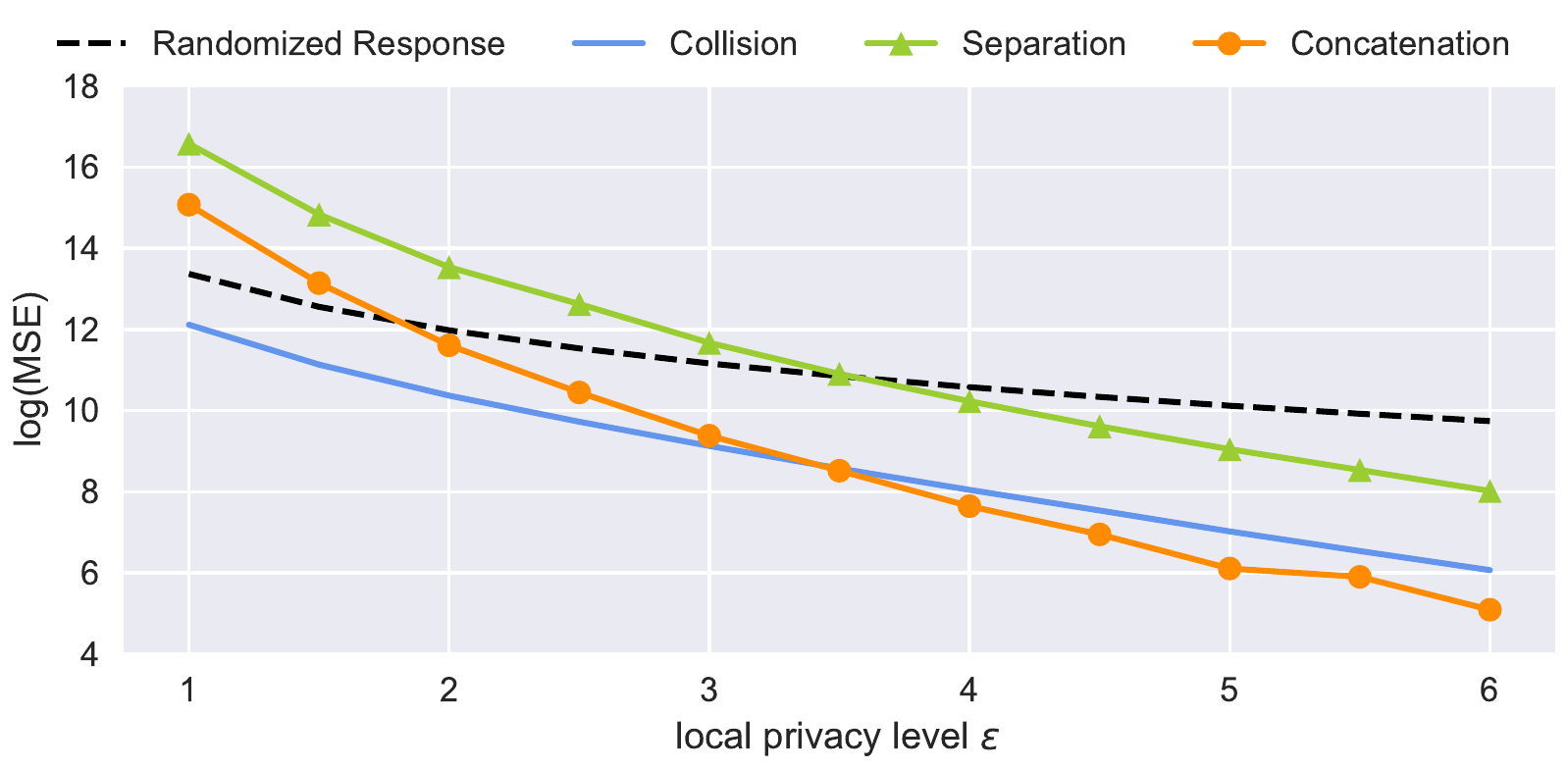}}
\vskip -0.1in
\caption{Comparison of local DP mechanisms in the low privacy regime for the BSVS problem with $n=1$, $s=200$, $|\mathcal{Y}|=50$, $k=2$, and $r=2$.}
\label{fig:approaches}
\end{center}
\vspace*{-2.3em}
\end{figure}

In Figure \ref{fig:approaches}, we compare the mean squared error of $(s\cdot |\mathcal{Y}|,k\cdot r,\epsilon,2k\cdot r-1+k\cdot r\cdot e^\epsilon)$-collision mechanism with the Separation/Concatenation approaches. It is observed that the Concatenation approach is more accurate when privacy budget is high (e.g., when $\epsilon\geq 3.5$), and the Concatenation dominates the Separation in all cases (due to smaller constant factor in error bounds).

\subsection*{C. Beyond Reverse $k$-NN Connection}
Recall that after privatization, the accuracy of a labeling answer $a_t$ highly relies on the non-private count gap $Gap_t$ between the true label and false labels. 
In order to increase the gap, it is desirable to assign more (similar) records to the $t$-th query. However, when utilizing the reverse $k$-NN rule for connecting records/queries, the degrees of queries might be imbalanced and some queries may only connect with few records, especially when domain shifts \cite{luo2019taking}. Hence, we investigate more balanced connections.



Let a bi-partied graph $\mathcal{G}=(Q, D, E)$ denote the connection relation between queries $Q$ and records $D$, where $E \subseteq Q\times D$ is the set of edges/connections. Since the maximum degree of nodes in $D$ is associated with privatization parameters, we here only consider connections $E \in \mathcal{E}^k$ that the nodes in $D$ have the maximum degree of $k$. Let $sim(q,x) \in \mathbb{R}^+$ denote the similarity between a node $q \in Q$ and a node $x \in D$ (derived from the latent representation of $q$ and $x$). Now for approximating the gap $Gap_t$, we define the score of a query $q$ given connections $E$ as:
$$Score(q) = \sum_{x\in D\ \text{and}\  (q,x)\in E} sim(q,x).$$
The reverse $k$-NN rule is thus equivalent to maximizing the (arithmetic) mean of these scores:
$\arg \max_{E \in \mathcal{E}^k} \frac{\sum_{q \in Q} Score(q)}{|Q|}.$
Alternatively, we may seek for more balanced connections by maximizing the minimum:
$\arg \max_{E \in \mathcal{E}^k} \min_{q \in Q} Score(q);$
or by maximizing the harmonic mean:
$\arg \max_{E \in \mathcal{E}^k} \frac{|Q|}{\sum_{q \in Q} 1/Score(q)}.$